\documentclass{article}
\usepackage[T1]{fontenc}
\usepackage{amssymb}
\usepackage[misc,geometry]{ifsym}
\usepackage{graphicx}%
\usepackage{multirow}%
\usepackage{mathrsfs}%
\usepackage[title]{appendix}%
\usepackage{xcolor}%
\usepackage{textcomp}%
\usepackage{manyfoot}%
\usepackage{booktabs}%
\usepackage{algorithm}%
\usepackage{algorithmicx}%
\usepackage{algpseudocode}%
\usepackage{listings}%
\usepackage{amsthm}

\newtheorem{theorem}{Theorem}
\newtheorem{corollary}{Corollary}[theorem]
\newtheorem{lemma}[theorem]{Lemma}
\newtheorem{definition}{Definition}[section]

\newtheorem{prop}[theorem]{Proposition}

\newcommand\Ga{\mathcal{G}^\mathrm{add}}
\newcommand\Ge{\mathcal{G}^\mathrm{epex}}
\newcommand\GG{\mathcal{G}}

\begin{document}

\title{Weighted Clique and Independent Set in Edge-Distant Hereditary Graphs}
\author{
	Eshwar Srinivasan and Ramesh Hariharasubramanian \footnote{Corresponding Author}\\
	{\small Department of Mathematics,
	Indian Institute of Technology Guwahati,
	Assam - 781039, India}\\
	{\small s.eshwar@iitg.ac.in and ramesh\_h@iitg.ac.in}
}
	
	\maketitle
\begin{abstract}
    	In this work, we investigate the algorithmic aspects of two natural extensions of hereditary classes: the edge-apex class and the edge-add class, recently introduced by Singh and Sivaraman. These are defined as the graph classes obtained by at most one edge deletion or one non-edge addition, respectively, from a hereditary class~$\GG$. Building on earlier results showing that both classes remain hereditary and admit finite forbidden induced subgraph characterizations whenever $\GG$ does, we focus on the \emph{Weighted Maximum Clique Problem} (WMCP) and the \emph{Weighted Maximum Independent Set Problem} (WMISP).
		
		We first present algorithms for WMCP and WMISP on both the edge-apex and edge-add classes of hereditary graph classes. Extending this framework, we introduce the notion of the \emph{$\GG$-edge distance} of a graph~$G$, denoted by $\xi_{\GG}(G)$, which quantifies how far $G$ is from the class~$\GG$ in terms of the minimum number of edge deletions or non-edge additions needed to transform it into a member of~$\GG$.
		
		By parameterizing with respect to this distance, we show that both WMCP and WMISP can be solved in $O^*(2^k)$ time on graphs whose $\GG$-edge distance is~$k$, provided these problems admit polynomial-time algorithms within the class~$\GG$. This result extends earlier algorithmic characterizations of the single edge-apex and edge-add classes to the more general setting of $k$-edge-distant graphs. By combining our general results with known properties of transitive graphs, we show that WMCP and WMISP can be solved in $O^*(2^k)$ time for graphs with transitive-edge distance $k$.
		
	\end{abstract}

\section{Introduction}\label{sec1}

The study of hereditary graph classes is particularly interesting. Many algorithmic problems that are hard to solve in general become tractable when restricted to certain hereditary graph classes. However, not every graph class is hereditary. It is well known that a graph class is hereditary if and only if it can be characterized by minimal forbidden induced subgraphs. In \cite{MR3772387}, Borowiecki, Drgas-Burchardt, and Sidorowicz studied the $k$-apex class of a hereditary graph class $\GG$, consisting of graphs that are at most $k$ vertex deletions away from $\GG$. In \cite{MR4794045,singh2025hereditary}, Singh and Sivaraman extended this concept to the $p$-edge-apex and $q$-edge-add classes of a hereditary class $\GG$, consisting of graphs that are at most $p$ edge deletions and $q$ non-edge additions away from $\GG$, respectively. They showed that the $p$-edge-apex and $q$-edge-add classes of a hereditary graph class $\GG$ are themselves hereditary, and proved that they have finite forbidden induced subgraphs whenever $\GG$ has finite forbidden induced subgraphs.

The \emph{weighted maximum clique problem} (respectively, the \emph{weighted maximum independent set problem}) is to determine a clique (respectively, an independent set) of maximum weight in a general weighted graph. There is a long line of research on both exact algorithms and approximation algorithms for these problems. 
On the one hand, exact algorithms typically rely on exponential-time methods such as branch-and-bound, dynamic programming over subsets, and inclusion--exclusion techniques, which aim to improve the running time over brute-force enumeration. 
On the other hand, approximation algorithms and heuristic methods, including semidefinite programming relaxations and greedy strategies, have been extensively studied, since both problems are NP-hard and unlikely to admit polynomial-time exact solutions.

The maximum clique problem (MCP) has long been a cornerstone of research in graph theory and combinatorial optimization, attracting sustained attention over several decades. One of the earliest influential contributions came from Carraghan and Pardalos~\cite{CARRAGHAN1990375}, who proposed a simple yet effective branch-and-bound (B\&B) algorithm that became the starting point for much of the work that followed. Since then, exact algorithms for MCP have steadily evolved---moving from basic B\&B schemes to more sophisticated hybrids that draw on techniques such as graph coloring, filtering, and satisfiability solving. This evolution not only pushed the limits of solving benchmark instances but also deepened our understanding of how to design algorithms for some of the most challenging problems in combinatorial optimization. For a comprehensive survey of exact algorithms for MCP, we refer the reader to~\cite{MR3293361}.

Let $G$ be a graph. Throughout this study, all graphs are assumed to be simple, that is, finite, undirected, and without loops or multiple edges. The vertex set and edge set of $G$ are denoted by $V(G)$ and $E(G)$, respectively. For a vertex $v \in V(G)$, the \emph{neighbourhood} of $v$ in $G$, denoted by $N_G(v)$, is the set of vertices adjacent to $v$. The \emph{degree} of $v$, denoted by $\deg_G(v)$, is the cardinality of $N_G(v)$. We denote by $\Delta(G)$ (respectively, $\delta(G)$) the maximum (respectively, minimum) degree of a vertex in $G$. A \emph{clique} of $G$ is a set of pairwise adjacent vertices, while an \emph{independent set} is a set of pairwise non-adjacent vertices.

For $X \subseteq V(G)$, the subgraph of $G$ induced by $X$, denoted $G[X]$, is the graph with vertex set $X$ and edge set containing all edges of $G$ whose endpoints lie in $X$. A graph $H$ is an induced subgraph of $G$ if $H$ is isomorphic to some induced subgraph of $G$. A graph class $\GG$ is said to be \emph{hereditary} if, for every $G \in \GG$, all induced subgraphs of $G$ also belong to $\GG$.

The \emph{complement} of $G$, denoted by $\overline{G}$, is the graph with vertex set $V(G)$, where two vertices $x,y \in V(G)$ form an edge in $\overline{G}$ if and only if $xy \notin E(G)$. The edges of $\overline{G}$ are called the \emph{non-edges} of $G$. For an edge $xy \in E(G)$, let $G - xy$ denote the graph obtained by deleting $xy$. For a non-edge $xy \in E(\overline{G})$, let $G + xy$ denote the graph obtained by adding $xy$. More generally, for $S \subseteq E(G)$ (respectively, $S \subseteq E(\overline{G})$), we write $G - S$ (respectively, $G + S$) for the graph obtained from $G$ by deleting (respectively, adding) all edges in $S$. For additional graph-theoretic terminology and notation, we refer the reader to \cite{MR1367739}.

The paper is organized as follows. In Section~\ref{s2}, we present algorithms for the weighted maximum clique problem (WMCP) and the weighted maximum independent set problem (WMISP) when restricted to the edge-apex and edge-add classes of a hereditary graph class~$\GG$. We also introduce a parameter called the \emph{$\GG$-edge distance} of a graph~$G$, denoted by $\xi_\GG(G)$, which we use to design algorithms for WMCP and WMISP on graphs whose $\GG$-edge distance is~$k$. In Section~\ref{s3}, building on the results of Section~\ref{s2}, we provide an $O^*(2^k)$-time algorithm---where $O^*(\cdot)$ is $O(\cdot)$ ignoring polynomial factors---for WMCP and WMISP when restricted to graphs whose \emph{transitive-edge distance} is~$k$.

\section{Weighted maximum clique \& weighted independent set of edge-apex class and edge-add class}\label{s2}

This section presents algorithms for the weighted maximum clique and weighted independent set problems on the edge-apex and edge-add classes of a hereditary class of graphs. Furthermore, we introduce a graph parameter for a given hereditary class $\GG$, called the $\GG$-edge-distance, which will be used in subsequent sections as a parameter for FPT algorithms.

Before presenting the results of this section, we recall some known results on the edge-apex and edge-add classes of a hereditary class of graphs.
\begin{definition}[\cite{MR4794045,singh2025hereditary}]
	The \textit{edge-add class} (respectively, \textit{edge-apex class}), denoted by $\Ga$ (respectively, $\Ge$), of a hereditary class $\GG$ of graphs is the class of graphs $G$ such that either $G \in \GG$ or $G$ has a non-edge (respectively, an edge) $xy$ for which $G + xy$ (respectively, $G - xy$) belongs to $\GG$.
\end{definition}

\begin{lemma}[\cite{singh2025hereditary}]
	Let $\GG$ be a hereditary class of graphs. Then both $\Ge$ and $\Ga$ are hereditary classes.
\end{lemma}

\begin{theorem}\label{t1}
	Let $G \in \Ga$ be a graph with $n$ vertices and $m$ edges, and let $xy \notin E(G)$ be such that $G + xy \in \GG$. If there exists an algorithm that computes a weighted maximum clique of $G + xy$ in $O(f(m,n))$ time, where $f$ is a non-decreasing function with $f(m,n) \in \Omega(m+n)$, then a weighted maximum clique of $G$ can also be found in $O(f(m,n))$ time.
\end{theorem}
\begin{proof}
	Let $G$ be a graph in $\Ga$ with $n$ vertices and $m$ edges. If $G \in \GG$, then the proof is complete. Otherwise, we assume that $G \notin \GG$ but $G$ has a non-edge $xy$ such that $G + xy \in \GG$.
	
	Let $\mathcal{A}(G', w_{G'})$ be an algorithm that takes as input a graph $G' \in \GG$ together with a weight function $w_{G'} : V(G') \rightarrow \mathbb{R}$, and outputs a clique $K' \subseteq G'$ of maximum weight and its cumulative weight $W(K')$ in $O(f(m', n'))$ time, where $n' = |V(G')|$, $m' = |E(G')|$, and $f$ is a real-valued non-decreasing function. We claim that the algorithm $\mathcal{A}^{add}(G, xy, w(G))$, given as input a graph $G \in \Ga$, a non-edge $xy$ of $G$ with $G + xy \in \GG$, and a weight function $w_G : V(G) \rightarrow \mathbb{R}$, outputs a weighted maximum clique $K$ of $G$ and its cumulative weight $W(K)$.
	
	Firstly, since $\GG$ is hereditary and $G - x$ and $G - y$ are induced subgraphs of $G + xy$, it follows that $G - x, G - y \in \GG$. Moreover, lines~1--2 of Algorithm~\ref{algo1} indicate that $K_1$ is the weighted maximum clique of $G - x$ with cumulative weight $W(K_1)$ and $K_2$ is the weighted maximum clique of $G - y$ with cumulative weight $W(K_2)$. Since $G - x$ and $G - y$ are induced subgraphs of $G$, it follows that $K_1$ and $K_2$ are cliques of $G$ as well. Lines~3--5 return the weighted maximum clique of $G$ along with its cumulative weight. 
	
	Suppose that $W(K_1) \ge W(K_2)$. Then it follows that $K = K_1$. For the sake of contradiction, assume that there exists a clique $K'$ of $G$, non-isomorphic to $K$, such that $W(K') > W(K)$. If $y \notin V(K')$, then $K'$ is a maximal clique of the graph $G - y$. Hence, from line~2, we have $W(K_2) \ge W(K')$, which implies $W(K_2) > W(K_1)$, a contradiction. Similarly, if $x \notin V(K')$, then $K'$ is a maximal clique of the graph $G - x$. From line~1, it follows that $W(K') \le W(K_1) = W(K)$, which is again a contradiction. Therefore, both $x, y \in V(K')$. Since $x$ is not adjacent to $y$ in $G$, it follows that $K'$ is not a clique, yielding a contradiction. Hence, $K$ is a weighted maximum clique of $G$. A similar argument applies to the case where $W(K_2) > W(K_1)$ by simply swapping $K_1$ with $K_2$ and $x$ with $y$.
	
	Now let us analyze the running time of Algorithm~\ref{algo1}. Based on the assumption in the theorem, and since $f$ is a non-decreasing function, lines~1--2 take $O(f(m, n))$ time. The remaining steps of the algorithm are executed in constant time. Therefore, Algorithm~\ref{algo1} runs in $O(f(m, n))$ time. This completes the proof of the theorem.
	\end{proof}

\begin{algorithm}
	\caption{$\mathcal{A}^{add}(G, xy, w_G)$}\label{algo1}
	\begin{algorithmic}[1]
		\Require A graph $G \in \Ga$, a non-edge $xy$ such that $G+xy \in \GG$, and a weight function $w_G \colon V(G) \rightarrow \mathbb{R}$.
		\Ensure A weighted maximum clique $K$ of cumulative weight $W(K)$.
		\State $(K_1, W(K_1)) \leftarrow \mathcal{A}(G - x, w_{G - x})$.
		\State $(K_2, W(K_2)) \leftarrow \mathcal{A}(G-y, w_{G - y})$.
		\If{$W(K_1) \ge W(K_2)$}\label{algln2}
		\Return $(K_1, W(K_1))$.
		\Else~~
		\Return $(K_2, W(K_2))$.
		\EndIf
	\end{algorithmic}
\end{algorithm}
\begin{theorem}\label{t2}
	Let $G \in \Ge$ be a graph with $n$ vertices and $m$ edges, and let $xy \in E(G)$ be such that $G - xy \in \GG$. If there exists an algorithm that computes a weighted maximum clique of $G - xy$ in $O(f(m,n))$ time, where $f$ is a non-decreasing function with $f(m,n) \in \Omega(m+n)$, then a weighted maximum clique of $G$ can also be found in $O(f(m,n))$ time.
\end{theorem}
\begin{proof}
	Let $G$ be a graph in $\Ge$ with $n$ vertices and $m$ edges. If $G \in \GG$, then the proof is complete. Otherwise, we assume that $G \not \in \GG$ and $G$ has an edge $xy$ such that $G - xy \in \GG$.
	
	Let $\mathcal{A}(G', w_{G'})$ be an algorithm that takes as input a graph $G' \in \GG$ together with a weight function $w_{G'} \colon V(G') \to \mathbb{R}$, and outputs a clique $K' \subseteq G'$ of maximum weight along with its cumulative weight $W(K')$, in $O(f(m', n'))$ time, where $n' = |V(G')|$, $m' = |E(G')|$, and $f$ is a real-valued non-decreasing function. We claim that the algorithm $\mathcal{A}^{\mathrm{epex}}(G, xy, w_G)$, given as input a graph $G \in \Ge$, an edge $xy$ of $G$ with $G - xy \in \GG$, and a weight function $w_G \colon V(G) \to \mathbb{R}$, outputs a weighted maximum clique $K$ of $G$ together with its cumulative weight $W(K)$.
	
	Since $\GG$ is hereditary, we have $G' \cong G[N_G(x) \cap N_G(y)] \in \GG$. Line~2 of Algorithm~\ref{algo2} computes the weighted maximum clique $K_1$ of the graph $G'$ along with its cumulative weight $W(K_1)$, and line~3 computes the weighted maximum clique $K_2$ of $G - xy$ along with its cumulative weight $W(K_2)$. Finally, lines~4--6 return the weighted maximum clique of $G$ together with its cumulative weight.
	
	Suppose that $W(K_1) + w(x) + w(y) \ge W(K_2)$. Then it follows that $K \cong G[V(K_1) \cup \{x, y\}].$ Assume, for the sake of contradiction, that there exists a clique $K'$ of $G$, non-isomorphic to $K$, such that $W(K') > W(K)$. If $xy \notin E(K')$, then $K'$ is a clique of $G - xy$. Therefore, by line~3 of Algorithm~\ref{algo2}, we have $W(K') \le W(K_2) \le W(K_1) + w(x) + w(y) = W(K),$ a contradiction. If $xy \in E(K')$, then $V(K') \setminus \{x, y\} \subseteq N_G(x) \cap N_G(y)$. Thus, $G[V(K' \setminus \{x, y\})]$ is a clique of $G'$. By line~2 of Algorithm~\ref{algo2}, it follows that $W(K') \le W(K_1) + w(x) + w(y) = W(K),$ again a contradiction. Hence, $K$ is a weighted maximum clique of $G$.  
	
	Suppose that $W(K_1) + w(x) + w(y) < W(K_2)$. Then it follows that $K \cong K_2.$ Since $K$ is a weighted maximum clique of $G - xy$, we have $xy \notin E(K)$. Assume, for the sake of contradiction, that there exists a clique $K'$ of $G$, non-isomorphic to $K$, such that $W(K') > W(K)$. If $xy \notin E(K')$, then $K'$ is a clique of $G - xy$. Therefore, by line~3 of Algorithm~\ref{algo2}, $W(K') \le W(K_2) = W(K),$ a contradiction. If $xy \in E(K')$, then $V(K') \setminus \{x, y\} \subseteq N_G(x) \cap N_G(y)$. Thus, $G[V(K' \setminus \{x, y\})]$ is a clique of $G'$. By line~2 of Algorithm~\ref{algo2}, it follows that $W(K') \le W(K_1) + w(x) + w(y) < W(K_2) = W(K),$ again a contradiction. Hence, $K$ is a weighted maximum clique of $G$.
	
	Now let us analyze the running time of Algorithm~\ref{algo2}. Based on the assumption in the theorem, and since $f$ is a non-decreasing function with $f(m,n)\in\Omega(m+n)$, lines~1--3 take $O(f(m,n))$ time. The remaining steps of the algorithm are executed in constant time. Therefore, Algorithm~\ref{algo2} runs in $O(f(m,n))$ time. This completes the proof of the theorem.
\end{proof}
\begin{algorithm}
	\caption{$\mathcal{A}^{epex}(G, xy, w_G)$}\label{algo2}
	\begin{algorithmic}[1]
		\Require A graph $G \in \Ge$, an edge $xy$ such that $G-xy \in \GG$, and a weight function $w_G \colon V(G) \rightarrow \mathbb{R}$.
		\Ensure A weighted maximum clique $K$ of cumulative weight $W(K)$.
		\State $G' \leftarrow G[N_G(x) \cap N_G(y)]$.
		\State $(K_1, W(K_1)) \leftarrow \mathcal{A}(G', w_{G'})$.
		\State $(K_2, W(K_2)) \leftarrow \mathcal{A}(G-xy, w_{G - xy})$.
		\If{$W(K_1) + w(x) + w(y) \ge W(K_2)$}\label{algln3}
		\Return $(G[V(K_1)\cup\{x, y\}], W(K_1) + w(x)+w(y))$.
		\Else~~
		\Return $(K_2, W(K_2))$.
		\EndIf
	\end{algorithmic}
\end{algorithm}

Let $\GG$ be a hereditary class of graphs. We say that a graph $G$ is in $\overline{\GG}$ if and only if $\overline{G} \in \GG$. Observe that $\GG$ is hereditary if and only if $\overline{\GG}$ is hereditary. By the definition of the edge-add class, a graph $G \in \overline{\GG}^\mathrm{add}$ if either $\overline{G} \in \GG$ or $G$ has a non-edge $xy$ such that $\overline{G} - xy \in \GG$. Similarly, by the definition of the edge-apex class, a graph $G \in \overline{\Ge}$ if either $\overline{G} \in \GG$ or $G$ has a non-edge $xy$ such that $\overline{G} - xy \in \GG$. This implies that $\overline{\GG}^\mathrm{add} = \overline{\Ge}$. A similar argument shows that $\overline{\GG}^\mathrm{epex} = \overline{\Ga}$. As a result, computing a weighted maximum independent set of a graph $G \in \Ge$ reduces to computing a weighted maximum clique of $\overline{G} \in \overline{\GG}^\mathrm{add}$, and computing a weighted maximum independent set of a graph $G \in \Ga$ reduces to computing a weighted maximum clique of $\overline{G} \in \overline{\GG}^\mathrm{epex}$, respectively. Therefore, the following corollaries follow directly from \textit{Theorem}~\ref{t1} and \textit{Theorem}~\ref{t2}.

\begin{corollary}
	Let $G \in \Ga$ be a graph with $n$ vertices and $m$ edges, and let $xy \notin E(G)$ be such that $G + xy \in \GG$. If there exists an algorithm that computes a weighted maximum independent set of $G + xy$ in $O(f(m,n))$ time, where $f$ is a non-decreasing function with $f(m,n) \in \Omega(m+n)$, then a weighted maximum independent set of $G$ can also be found in $O(f(m,n))$ time.
\end{corollary}
\begin{corollary}
	Let $G \in \Ge$ be a graph with $n$ vertices and $m$ edges, and let $xy \in E(G)$ be such that $G - xy \in \GG$. If there exists an algorithm that computes a weighted maximum independent set of $G - xy$ in $O(f(m,n))$ time, where $f$ is a non-decreasing function with $f(m,n) \in \Omega(m+n)$, then a weighted maximum independent set of $G$ can also be found in $O(f(m,n))$ time.
\end{corollary}
\begin{definition}\cite{MR4794045,singh2025hereditary}
	For a fixed integer $p \ge 0$, the \textit{$p$-edge-add class} (respectively, \textit{$p$-edge-apex class}), denoted by $\Ga_p$ (respectively, $\Ge_p$), of a hereditary class $\GG$ of graphs is the class of graphs $G$ such that there exists a set (possibly empty) of non-edges $S \subseteq E(\overline{G})$ (respectively, edges $S \subseteq E(G)$) with $|S| \le p$ such that $G + S \in \GG$ (respectively, $G - S \in \GG$).
\end{definition}

\begin{lemma}[\cite{singh2025hereditary}]\label{l2}
Let $\GG$ be a hereditary class of graphs. Then, for non-negative integers $p$ and $q$, both $\Ge_p$ and $\Ga_q$ are hereditary classes.
\end{lemma}

For any two non-negative integers $p$ and $q$, a graph $G$ may belong simultaneously to the $p$-edge-add class and the $q$-edge-apex class of a hereditary class $\GG$ of graphs. Hence, we introduce the following parameter for graphs.
\begin{definition}
Let $G$ be a graph and $\GG$ a hereditary class of graphs. For a non-negative integer $k$, we say that $G$ is \emph{$k$-edge-distant} from $\GG$ if $G$ belongs to either $\Ge_k$ or $\Ga_k$. Moreover, the minimum $k$ for which $G$ is $k$-edge-distant is called the \emph{$\GG$-edge distance} of $G$, denoted by $\xi_{\GG}(G)$. A set $S$ of edges (respectively, non-edges) such that $G-S$ (respectively, $G+S$) belongs to $\GG$ is called a \emph{$\GG$-distant-edge set}.
\end{definition}

For any hereditary graph class $\GG$, the following theorem generalizes \textit{Theorem}~\ref{t1} and \textit{Theorem}~\ref{t2} to graphs with $\GG$-edge distance greater than 1.

\begin{theorem}\label{t3}
	Let $G$ be a graph with $n$ vertices, $m$ edges, and let $\GG$ be a hereditary graph class. Let the $\GG$-edge distance of $G$ be $k$, that is, $\xi_\GG(G) = k$. If the weighted maximum clique problem when restricted to the hereditary class of graphs $\GG$ can be computed in $O(f(m, n))$ time, where $f$ is a non-decreasing function with $f(m, n) \in \Omega(m + n)$, then a weighted maximum clique of $G$ can be found in $O(2^k f(m, n))$ time.
\end{theorem}
\begin{proof}
	Let $G$ be a graph with $n$ vertices and $m$ edges, and let $\GG$ be a hereditary graph class. Suppose $\xi_{\GG}(G) = k$. Then either there exists a set $S \subseteq E(G)$ with $|S| \leq k$ such that $G - S \in \GG$, or there exists a set $S' \subseteq E(\overline{G})$ with $|S'| \leq k$ such that $G + S' \in \GG$. For any graph $H \in \GG$, let $\mathcal{A}(H, w_H)$ be an algorithm that takes as input the graph $H$ together with a weight function $w_H \colon V(H) \rightarrow \mathbb{R}$, and outputs a weighted maximum clique of $H$ in $T(m', n')$ time, where $n' = |V(H)|$, $m' = |E(H)|$, and $T(m', n') = O(f(m', n'))$ for some non-decreasing function $f$ with $f(m', n') \in \Omega(m'+ n')$.
	
	\textbf{Case 1: }Suppose that there exists a set $S \subseteq E(G)$ with $|S| \leq k$ such that $G - S \in \GG$. This implies that $G \in \Ge_k$. We claim that the algorithm $\mathcal{A}^{\mathrm{epex}}_k(G, S, w_G)$, given as input a graph $G \in \Ge_k$, a set $S \subseteq E(G)$ with $|S| \le k$ such that $G - S \in \GG$, and a weight function $w_G \colon V(G) \to \mathbb{R}$, outputs a weighted maximum clique $K$ of $G$ together with its cumulative weight $W(K)$. The proof of this claim proceeds by induction on $k$.
	
	Consider the base case $k = 1$. This implies that $G \in \Ge$. From the proof of \textit{Theorem}~\ref{t2}, it follows that Algorithm~\ref{algo2}, when given as input the graph $G$ together with a set $S \subseteq E(G)$ (which in this case is $S = \{xy\}$) and a weight function $w_G$, outputs a weighted maximum clique of $G$ along with its cumulative weight.
	Therefore, $\mathcal{A}^\mathrm{epex}_1(G, S, w_G)$ coincides with $\mathcal{A}^\mathrm{epex}(G, xy, w_G)$ of Algorithm~\ref{algo2}. Moreover, the running time of Algorithm~\ref{algo2} can be verified to be $2T(m, n) + O(1)$. Now assume that for $k > 1$, there exists an algorithm $\mathcal{A}^\mathrm{epex}_{k-1}(H, S', w_{H})$ which, given as input a graph $H \in \Ge_{k-1}$, a set $S' \subseteq E(H)$ with $|S'| \le k-1$ such that $H - S' \in \GG$, and a weight function $w_{H}$, outputs a weighted maximum clique of $H$ along with its cumulative weight in $2^{k-1}T(m, n) + (k - 1)O(1)$ time. This serves as the induction hypothesis.
	
	Since $G \in \Ge_k$ and $G - S \in \GG$, there exists an edge $xy \in S$ such that $G - xy \in \Ge_{k-1}$. We claim that Algorithm~\ref{algo4}, which is a slight modification of Algorithm~\ref{algo2}, serves as the algorithm $\mathcal{A}^\mathrm{epex}_k(G, S, w_G)$. By \textit{Lemma}~\ref{l2}, it follows that $\Ge_{k-1}$ is hereditary. Since $\Ge_{k}$ is an edge-apex class of $\Ge_{k-1}$, a similar argument to the proof of correctness of Algorithm~\ref{algo2} establishes the correctness of Algorithm~\ref{algo4}.
	
	By the induction hypothesis, $\mathcal{A}^\mathrm{epex}_{k-1}(H, S', w_H)$ outputs a weighted maximum clique of $H$ along with its cumulative weight  in $2^{k-1}T(m', n') + (2^{k-1}-1)O(1)$ time, where $H \in \Ge_{k-1}$, $n' = V(H)$, and $m' = E(H)$.
	
	Algorithm~\ref{algo4} computes the weighted maximum clique of $G$ by considering the two cases: 
	(i) computing the weighted maximum clique of $G' \cong G[N_G(x) \cap N_G(y)]$ along with its cumulative weight, and 
	(ii) computing the weighted maximum clique of $G - xy$ along with its cumulative weight. 
	
	Both of these steps recursively compute the weighted maximum clique and its cumulative weight. Since the remaining operations outside the recursive calls take $O(1)$ time and $T(m, n) = O(f(m, n))$, it follows that Algorithm~\ref{algo4} runs in $O(2^k f(m, n))$ time.
	
	\textbf{Case 2: }Suppose that there exists a set $S \subseteq E(\overline{G})$ with $|S| \leq k$ such that $G + S \in \GG$. This implies that $G \in \Ga_k$. We claim that the algorithm $\mathcal{A}^{\mathrm{add}}_k(G, S, w_G)$, given as input a graph $G \in \Ga_k$, a set $S \subseteq E(\overline{G})$ with $|S| \le k$ such that $G + S \in \GG$, and a weight function $w_G \colon V(G) \to \mathbb{R}$, outputs a weighted maximum clique $K$ of $G$ together with its cumulative weight $W(K)$. 
	
	A similar argument to the previous case shows that Algorithm~\ref{algo3}, which is a slight modification of Algorithm~\ref{algo1}, serves as the algorithm $\mathcal{A}^\mathrm{add}_k(G, S, w_G)$. By \textit{Lemma}~\ref{l2}, it follows that $\Ga_{k-1}$ is hereditary. Since $\Ga_k$ is an edge-add class of $\Ga_{k-1}$, a similar argument to the proof of correctness of Algorithm~\ref{algo1} establishes the correctness of Algorithm~\ref{algo3}.
	
	We consider a similar induction hypothesis as in the previous case. Since the running time of Algorithm~\ref{algo1} is $2T(m, n) + O(1)$, the algorithm $\mathcal{A}^\mathrm{add}_{k-1}(H, S', w_H)$ outputs a weighted maximum clique of $H$ along with its cumulative weight in $2^{k-1}T(m', n') + (2^{k-1}-1)O(1)$ time, where $H \in \Ga_{k-1}$, $n' = |V(H)|$, and $m' = |E(H)|$.
	
	Since $G \in \Ga_k$ and $G + S \in \GG$, there exists a non-edge $xy \in S$ such that $G + xy \in \Ga_{k-1}$. Algorithm~\ref{algo3} computes the weighted maximum clique of $G$ by recursively computing the weighted maximum cliques of $G - x$ and $G - y$ along with their cumulative weights. Since the remaining operations outside the recursive calls take $O(1)$ time and $T(m, n) = O(f(m, n))$, it follows that Algorithm~\ref{algo3} runs in $O(2^k f(m, n))$ time.
\end{proof}
\begin{algorithm}
	\caption{$\mathcal{A}^{epex}_{k}(G, S, w_G)$}\label{algo4}
	\begin{algorithmic}[1]
		\Require A graph $G \in \Ge_k$, a set $S$ of edges of $G$ such that $G-S \in \GG$, and a weight function $w_G \colon V(G) \rightarrow \mathbb{R}$.
		\Ensure A weighted maximum clique $K$ of cumulative weight $W(K)$.
		\State $G' \leftarrow G[N_G(x) \cap N_G(y)]$.
		\State $(K_1, W(K_1)) \leftarrow \mathcal{A}^\mathrm{epex}_{k-1}(G', S\cap E(G'), w_{G'})$.
		\State $(K_2, W(K_2)) \leftarrow \mathcal{A}^\mathrm{epex}_{k-1}(G-xy, S\setminus\{xy\}, w_{G - xy})$.
		\If{$W(K_1) + w(x) + w(y) \ge W(K_2)$}
		\Return $(G[V(K_1)\cup\{x, y\}], W(K_1) + w(x)+w(y))$.
		\Else~~
		\Return $(K_2, W(K_2))$.
		\EndIf
	\end{algorithmic}
\end{algorithm}

\begin{algorithm}
	\caption{$\mathcal{A}^{add}_k(G, S, w_G)$}\label{algo3}
	\begin{algorithmic}[1]
		\Require A graph $G \in \Ga_k$, a set $S$ of non-edges of $G$ such that $G+S \in \GG$, and a weight function $w_G \colon V(G) \rightarrow \mathbb{R}$.
		\Ensure A weighted maximum clique $K$ of cumulative weight $W(K)$.
		\State $(K_1, W(K_1)) \leftarrow \mathcal{A}^\mathrm{add}_{k-1}(G - x, S \cap E(G-x), w_{G - x})$.
		\State $(K_2, W(K_2)) \leftarrow \mathcal{A}^\mathrm{add}_{k-1}(G-y, S\cap E(G - y), w_{G - y})$.
		\If{$W(K_1) \ge W(K_2)$}
		\Return $(K_1, W(K_1))$.
		\Else~~
		\Return $(K_2, W(K_2))$.
		\EndIf
	\end{algorithmic}
\end{algorithm}

In can be noted that for a non-negative integer $k$, $\overline{\Ga_k} = \overline{\GG}^\mathrm{epex}_k$ and $\overline{\Ge_k} = \overline{\GG}^\mathrm{add}_k$. Therefore, the following corollary is a direct consequence of \textit{Theorem}~\ref{t3}.
\begin{corollary}\label{c3}
		Let $G$ be a graph with $n$ vertices, $m$ edges, and let $\GG$ be a hereditary graph class. Let the $\GG$-edge distance of $G$ be $k$, that is, $\xi_\GG(G) = k$. If the weighted maximum independent set problem when restricted to the hereditary class of graphs $\GG$ can be computed in $O(f(m, n))$ time, where $f$ is a non-decreasing function with $f(m, n) \in \Omega(m + n)$, then a weighted maximum independent set of $G$ can be found in $O(2^k f(m, n))$ time.
\end{corollary}
\section{Edge-distant transitive graphs}\label{s3}
In this section, we apply the results of the previous section to the class of transitive graphs. As a consequence, we show that the weighted maximum clique problem and the weighted independent set problem can be solved in polynomial time for graphs with bounded transitive-edge-distance.

Before presenting the results on edge-distant transitive graphs, we first recall some known results on transitive graphs that will be used in this section.

\begin{theorem}[\cite{MR2063679}]
	Given an undirected graph $G$, both the recognition of transitive graphs and the computation of a transitive orientation can be performed in $O(\Delta m)$ time, where $\Delta = \Delta(G)$ and $m = |E(G)|$.
\end{theorem}

\begin{theorem}[\cite{MR2063679}]\label{t4}
	The weighted maximum clique problem can be solved in linear time with respect to the size of the graph when restricted to the class of transitive graphs. Moreover, the weighted maximum independent set problem can be solved in polynomial-time with respect to the size of the graph when restricted to the class of transitive graphs.
\end{theorem}

\begin{definition}
	For an integer $k \ge 0$, a graph $G$ is said to be \emph{$k$-edge-distant transitive} if either  
	\begin{enumerate}
		\item there exists a set $S \subseteq E(G)$ with $|S| \le k$ such that $G - S$ is transitive, or  
		\item there exists a set $S \subseteq E(\overline{G})$ with $|S| \le k$ such that $G + S$ is transitive.
	\end{enumerate}
	The minimum integer $k$ for which a graph $G$ is $k$-edge-distant transitive is called the \emph{transitive-edge-distance} of $G$, denoted by $\xi_{\tau}(G)$, where $\tau$ denotes the class of transitive graphs. The corresponding set $S$ of edges or non-edges is called the \emph{transitive distant-edge set}.
\end{definition}

From the forbidden induced subgraph characterization of transitive graphs given by Gallai~\cite{MR221974}, we observe the following fact.
\begin{prop}\label{p1}
	Let $G$ be a graph with $m$ edges and $\overline{m}$ non-edges. If $\min \{m, \overline{m}\} \le 4$, then $G$ is transitive. 
\end{prop}


\begin{theorem}
	Let $G$ be a graph with $m$ edges and $\overline{m}$ non-edges. Then the transitive-edge distance of $G$, $\xi_{\tau}(G) \le \min\{m-4, \overline{m} - 4\}$. Moreover, for any graph $G$ on $n$ vertices, $\xi_{\tau}(G) < (n^2-n)/4$.
\end{theorem}
\begin{proof}
Let $G$ be a graph on $n$ vertices, $m$ edges, and $\overline{m}$ non-edges. By \textit{Proposition}~\ref{p1}, it follows that 
\[
\xi_{\tau}(G) \leq \min\{m-4, \overline{m}-4\}.
\] 
Moreover, since 
\[
\min\{m-4, \overline{m}-4\} < \frac{n^2-n}{4},
\] 
we obtain 
\[
\xi_{\tau}(G) < \frac{n^2-n}{4}.
\]
\end{proof}

As a direct consequence of \textit{Theorem}\ref{t3}, \textit{Corollary}~\ref{c3} and \textit{Theorem}~\ref{t4}, we obtain the following theorem.
\begin{theorem}
Let $G$ be a graph with $n$ vertices, $m$ edges, and $\overline{m}$ non-edges. If $\xi_{\tau}(G) = k$, then there exists an algorithm that, given a graph $G$, a transitive distant-edge set $S$, and a real-valued weight function, outputs a weighted maximum clique (respectively, independent set) of $G$ in $O(2^k f(size(G)))$, where $f(size(G))$ is a linear (respectively, polynomial) function of the size of $G$.
\end{theorem}

\section{Conclusion and Future Work}
In this paper, we studied the Weighted Maximum Clique Problem (WMCP) and the Weighted Maximum Independent Set Problem (WMISP) on the edge-apex and edge-add classes of hereditary graph classes. A central concept in our work is the $\GG$-edge distance $\xi_{\GG}(G)$, which quantifies how far a graph is from a hereditary class in terms of edge deletions or additions. By parameterizing with respect to this distance, we established that both WMCP and WMISP can be solved in $O^*(2^k)$ time on graphs whose $\GG$-edge distance is~$k$, provided these problems are solvable in polynomial time within the class~$\GG$. This result yields fixed-parameter tractable algorithms for a wide range of hereditary graph classes, highlighting the power of structural distance as a tool in algorithm design.

A notable application of our framework lies in the class of transitive graphs. Since these graphs admit polynomial-time solutions for both WMCP and WMISP, our parameterized results directly imply tractability for graphs that are within bounded transitive-edge distance.

The contributions of this work open several avenues for future research. One natural direction is to investigate whether similar parameterizations based on edge or vertex modifications can be applied to other classical NP-hard problems, such as graph coloring, domination, or Hamiltonicity. Another promising direction is to study specific graph classes more closely and establish upper bounds on their transitive-edge distance.


\end{document}